\documentclass[11pt]{article}
\usepackage{amsmath}
\usepackage{amssymb}
\usepackage{amsthm,amsxtra}
\usepackage{epsf}
\usepackage{eepic}
\usepackage{graphicx}
\newlength{\mytopmargin}
\newlength{\myleftmargin}
\usepackage[square,sort,comma,numbers]{natbib}
\newtheorem{lemma}{Lemma}
\newtheorem{prop}[lemma]{Proposition}
\newtheorem{thm}{Theorem}

\theoremstyle{remark}
\newtheorem{remark}[thm]{Remark}

\numberwithin{equation}{section}

\begin{document}

\title{\Large\bf  A generalisation of the relation between
zeros of the complex Kac polynomial and eigenvalues of
truncated unitary matrices}
\author{Peter J. Forrester and Jesper R. Ipsen}
\date{}
\maketitle

\begin{center}
\it
School of Mathematics and Statistics, ARC Centre of Excellence for Mathematical
 and Statistical Frontiers, The University of Melbourne, \\
Victoria 3010, Australia
\end{center}

\bigskip
\begin{center}
\bf Abstract  
\end{center}
\par
\bigskip
\noindent 
{\small The zeros of the random Laurent series $1/\mu - \sum_{j=1}^\infty c_j/z^j$, where each $c_j$ is an independent standard complex
Gaussian, is known to correspond to  the scaled eigenvalues of a   particular additive rank 1 perturbation of a standard complex Gaussian matrix. 
For the corresponding random Maclaurin series obtained by the replacement  $z \mapsto 1/z$,
we show that these same zeros correspond to the scaled eigenvalues of a  particular
multiplicative rank 1 perturbation of a random unitary matrix. Since the correlation functions of the latter are known, by taking an appropriate limit the
correlation functions for the random Maclaurin series can be determined. Only for $|\mu| \to \infty$ is a determinantal point process
obtained. For the one and two point correlations, by  regarding the Maclaurin series as the limit of a random polynomial, a direct calculation can also
be given.
}

\section{Introduction}
The Kac polynomial refers to the random $N$-th degree polynomial with coefficients given by independent standard real Gaussians \cite{Ka43}.
The zeros form a two-component point process --- one component is the real zeros, and the other the complex zeros, with the latter
occurring in complex conjugate pairs. For finite $N$ the $k$-point correlation for the complex zeros is known \cite{Pr96} to be given
by a so-called semi-permanent (or Hafnian \cite{IKO05}) of a $2k \times 2k$ matrix, which bears the same relation to a permanent as a determinant does to
a Pfaffian. However the entries of the matrix depend on $k$, and only for $k=1$ have they been evaluated explicitly \cite{Pr96}.

This situation simplifies dramatically in the limit $N \to \infty$. Then the zeros form a two-component Pfaffian point process \cite{Fo10a,MS14}, with the underlying
(matrix) correlation kernel independent of the number of real and/ or complex eigenvalues being considered. Moreover, this Pfaffian point process is
identical to that formed by the real and complex eigenvalues of any $(N-1) \times (N-1)$ sub-block, formed by deleting any
one row and any one column, of a random real orthogonal matrix chosen with Haar measure in the limit $N \to \infty$ \cite{Fo10}.

The class of results just revised were first discovered for the complex version of the Kac polynomial, i.e.~the 
$n$-th degree polynomial with coefficients given by independent  standard complex Gaussians. Thus in the limit $n \to \infty$
 the zeros, which with probability one are complex and {\it do not} come in complex conjugate pairs, form a determinantal point
 process with correlation kernel \cite{PV03}
 \begin{equation}\label{1}
 K(w,z) = {1 \over \pi (1 - w \bar{z})^2}.
 \end{equation}
 Here it is assumed that the zeros have modulus less than one --- those with modulus greater than 1 are statistically independent.
 And moreover this determinantal point process is identical to that formed by the eigenvalues of 
 any $(N-1) \times (N-1)$ sub-block of a random complex unitary matrix chosen with Haar measure in the limit $N \to \infty$ \cite{Kr06}.
 For a recent application of this latter coincidence to the problem of persistence exponents, see \cite{PS18}.
 
 A different, and more general, coincidence between the distribution of the zeros of a random power series with coefficients independently distributed
 as standard real Gaussians,
 and that of the eigenvalues of a particular random matrix ensemble, has been given by Tao \cite{Ta12}. Let $X$ be an $N \times N$ matrix with standard
 real Gaussian entries, let ${\mathbf 1}_N$ denote a vector with all entries equal to 1, and let $\mathbf{g}$ denote an $N \times 1$ column vector of
 independent standard real Gaussian entries. It was shown in \cite{Ta12} (without the specialisation to Gaussian entries, and thus in more general
 circumstances) that the eigenvalues of the random matrix 
 \begin{equation}\label{Xg}
 {1 \over \sqrt{N}} \Big ( X + \mu {\mathbf 1}_N \mathbf{g}^T \Big ),
 \end{equation}
 itself being motivated by a model in the theory of neural networks \cite{RA06}, in the limit $N \to \infty$
 and for $N \to \infty$ , for $|z| > 1$ are given
 by the zeros of the random Laurent series
 \begin{equation}\label{L1a}
 {1\over \mu} - \sum_{j=1}^\infty {c_j \over z^j}.
 \end{equation} 
 Here each $c_j$ is an independent standard
 real Gaussian. Note that in the limit $|\mu| \to \infty$ this, in the variable $1/z = w$, reduces to seeking the zeros of the limiting Kac polynomial.
 Our interest in this paper is on the complex analogue of this result, and moreover on the specification of the corresponding point process. 
 
 It is noted in \cite[Remark 1.12]{Ta12} that choosing $X$ and $\mathbf{g}$ in (\ref{Xg}) to have independent standard complex Gaussian entries, rather than
 independent standard real Gaussian entries, the eigenvalues in the limit $N \to \infty$ and for $|z| > 1$ are given
 by the zeros of the random Laurent series (\ref{L1a}) with each $c_j$ an independent standard
 complex Gaussian. Our first result, Proposition \ref{P1} below, is to show that the random matrix $UA$, with $U$ chosen with Haar measure from
 the classical unitary group $U(N)$, and $A = {\rm diag} \, (a,1,\dots,1)$ has upon setting $a =1/ \mu \sqrt{N}$ and taking the limit $N \to \infty$ its
 eigenvalues given by the zeros of (\ref{L1a}), in the variable $w = 1/z$. This random matrix is a multiplicative rank 1 perturbation of $U$, and
 its eigenvalue probability density function (PDF) is known for finite $N$ \cite{Fy01,FS03}. Moreover, the general $k$-point correlation function can be
 computed explicitly; see Proposition \ref{P2}.
  It is given in terms of a sum of $k$ terms each dependent on a correlation kernel, although the point process itself is not 
 technically determinantal (i.e.~the $k$-point correlation is not given by the determinant of a correlation kernel which itself is independent of $k$.)
 In Section \ref{S3} we show how the correlation functions for the zeros of (\ref{L1a}) can be computed directly, and use this formalism to given
 independent derivations of the one and two point correlation functions.
 
 \section{The random matrix ensemble $UA$, for $U \in U(N)$, and $A = {\rm diag} \, (a,1,\dots,1)$}
 \subsection{Characterisation of the eigenvalues in the limit $N \to \infty$}
 Consider the random matrix $UA$, for  $U \in U(N)$ chosen with Haar measure (see e.g.~\cite[\S 5.2]{DF17} for a practical construction), and
 $A = {\rm diag} \, (a,1,\dots,1)$. For $\mathbf{\psi}$ a normalised eigenvector with eigenvalue $\lambda$, one has that $\mathbf{\psi}^\dagger A^\dagger A \mathbf{\psi} =
 | \lambda |^2$ and thus with $|a| < 1$ it must be that $|\lambda| < 1$. We assume henceforth that $|a| < 1$.  
  As remarked,  $UA$ is a multiplicative rank 1 perturbation of $U$.  
 It is a simple exercise to make use
 of this fact 
 to transform the characteristic polynomial for the eigenvalues of $UA$ to a form
 in which each coefficient can, in the limit $N \to \infty$, be determined explicitly.
 
 \begin{prop}\label{P1}
 Consider the random matrix $UA$ as specified above, and set $a = 1/ (\mu \sqrt{N})$.
 In the limit $N \to \infty$ the eigenvalues of $UA$ are given by the zeros of the random Laurent series (\ref{L1a})
 in the variable $\lambda = 1/z$, $|\lambda| < 1$, 
 with each $c_j$ an independent standard
 complex Gaussian.
 \end{prop}
 
 \begin{proof}
 Let $\mathbf{e}_1$ denote the $N \times 1$ column vector with 1 in the first entry, and 0's elsewhere, and 
 set $I_N' = {\rm diag} \,(0,1,\dots,1)$.
 Simple manipulation gives the rewrite for the characteristic polynomial of the matrix $UA$,
 \begin{equation}\label{Z1}
 \det ( \lambda  I_N - UA) =  (-1)^N \det U A \det ( I_N - \lambda U^\dagger A^{-1} ).
 \end{equation}
 Further simple manipulation shows that the final factor on the RHS of this expression can itself
 be factored according to
 \begin{equation}
 \det ( I_N - \lambda U^\dagger A^{-1} ) \\
 = \det (I_N - \lambda I_N' U^\dagger)
 \det \Big ( I_N - {\lambda \over a} (I_N - \lambda U^\dagger I_N')^{-1} U^\dagger \mathbf{e}_1 \mathbf{e}_1^T \Big ).
 \end{equation}
 Taking into consideration the above equalities, it follows that the condition for an eigenvalue of $UA$ can be written
 \begin{equation}\label{W}
   \det \Big ( I_N - {\lambda \over a} (I_N - \lambda U^\dagger I_N')^{-1} U^\dagger \mathbf{e}_1 \mathbf{e}_1^T \Big )  = 0.
\end{equation}

 Generally, for $C$ an $n \times m$ matrix and $D$ an $m \times n$ matrix, we have that (see e.g.~\cite[Exercises 5.2 q.2]{Fo10})
 $$
 \det (I_n + C D) = \det (I_m + D C).
 $$
 Applying this with $n=N$, $m=1$, $C= ( \lambda  I_N - UI_N')^{-1}\mathbf{e}_1$ and $D =  \mathbf{e}_1^T$ shows
 (\ref{W}) can be rewritten as the scalar equation
 $$
1 - {\lambda \over a}    \mathbf{e}_1^T  (  I_N -   \lambda U^\dagger I_N')^{-1}   U^\dagger \mathbf{e}_1 = 0.
$$
For $|\lambda| < 1$, the inverse can be expanded according to the geometric series, and we obtain
 \begin{equation}\label{Ue}
 0 = 1 - {\lambda \over a}  \sum_{k=0}^\infty \lambda^k  \mathbf{e}_1^T  (U^\dagger I_N')^k U^\dagger \mathbf{e}_1.
 \end{equation}
 
 Now it is a known result \cite{Kr06} that for $V \in U(N)$ chosen with Haar measure, in the limit $N \to \infty$
  \begin{equation}\label{Va} 
  \sqrt{N} ( \mathbf{e}_1^T V \mathbf{e}_1, \mathbf{e}_1^T V^2 \mathbf{e}_1,  \mathbf{e}_1^T V^3 \mathbf{e}_1, \dots )
  \mathop{=}^{\rm d} (\alpha_1,\alpha_2,\alpha_3,\dots ),
  \end{equation}
  where each $\alpha_i$ is an independent standard complex Gaussian. To make use of this, 
set $a = 1/ (\mu \sqrt{N})$ and  
  note that
 \begin{equation}\label{Va1} 
\mathbf{e}_1^T (U^\dagger I_N')^k U^\dagger \mathbf{e}_1 =
 \mathbf{e}_1^T (U^\dagger)^{k+1}  \mathbf{e}_1 \Big ( 1 + {\rm O}(k/N) \Big ).
  \end{equation}
 Up to the technical issue of tightness (see the discussion
 in \cite[proof of Th. 4.3.15, pg.~112]{HKPV08}), and its resolution in
 \cite{Kr06}) use of (\ref{Va1}) and (\ref{Va}) together
  in (\ref{Ue}) implies that in the limit $N \to \infty$ the RHS of the latter reduces to the stated random Laurent
  polynomial.
\end{proof}

\section{Correlations for the eigenvalues of the limiting random matrices $UA$}
\subsection{The eigenvalue density}
Let $G = {\rm diag} \, (g_1, g_2, \dots, g_N)$ with each $g_i \ge 0$.
With $U \in U(N)$ chosen with Haar measure, 
Wei and Fyodorov \cite{WF08} have given an explicit formula for the eigenvalue density of
$U \sqrt{G}$.
While this is quite complicated in general, in the special case $G = {\rm diag} \, (|a|^2, 1, \dots, 1)$
corresponding to the random matrix $UA$ as specified in the Introduction, this simplifies to read
 \cite[eq.(2.12)]{WF08}
 \begin{equation}\label{B1}
 \rho_{(1)}(z) = {(|z|^2 - |a|^2)^{N-2} \over (1 - |a|^2)^{N-1} |z|^{2N}}
 \Big ( (N-1) (|z|^{2N} + |a|^2) +
 \sum_{k=0}^{N-2} \Big ( (N-2-k) + k |a|^2 \Big ) |z|^{2(N-1-k)} \Big ),
 \end{equation}
 where $1 > |z| > |a|$.
 
 \begin{prop}
 Set $a = 1/(\mu \sqrt{N})$.
In the limit $N \to \infty$ we have
\begin{equation}\label{R1}
\rho_{(1)}(z) = {1 \over \pi} {1 \over (1 - |z|^2)^2} \exp \Big ( {1 \over  |\mu|^2} \Big ( 1- {1 \over |z|^2} \Big ) \Big )
\Big (1 +{1 \over  |\mu|^2  |z|^4} (1 - |z|^2) \Big ),
\end{equation}
restricted to $|z|<1$. 
\end{prop}

\begin{proof}
The restriction to $|z|<1$ follows immediately from the restriction on $|z|$ noted below (\ref{B1}).

The summation over $k$ in (\ref{B1}) can be evaluated according to
$$
|z|^2 (1 - |a|^2) {\partial \over \partial  |z|^2}
{1 - |z|^{2(N-1)} \over 1 - |z|^2} + |a|^2 (N-2) {|z|^2 (1 - |z|^{2(N-1)} ) \over 1 - |z|^2}.
$$
Hence with $|z| < 1$ and fixed, uniformly in $|a|<1$,
\begin{equation}\label{SR}
\rho_{(1)}(z) \mathop{\sim}_{N \to \infty}
{(|z|^2 - |a|^2)^{N-2} \over (1 - |a|^2)^N |z|^{2N}} \bigg (
N |a|^2 + |z|^4 (1 - |a|^2)
{\partial \over \partial |z|^2} {1 \over 1 - |z|^2} + |a|^2 N
{|z|^2 \over 1 - |z|^2} \bigg ).
\end{equation}
We write $(|z|^2 - |a|^2)^{N-2}  = |z|^{2(N-2)} (1 - |a|^2/|z|^2)^{N-2}$.
Upon this manipulation, and with $|a| =  1/(|\mu| \sqrt{N})$,
the elementary limit $(1-u/N)^N \to e^{-u}$ as
$N \to \infty$ establishes
(\ref{R1}).
\end{proof}

\begin{remark}\label{Re1}
We see from (\ref{R1}) that the density vanishes at the rate of an essential singularity as $|z| \to 0$. Another feature is
that as $|z| \to 1^-$, the leading form of $\rho_{(1)}(z)$ is independent of $\mu$.
\end{remark}

\subsection{The $k$-point correlation}
Prior to the derivation of (\ref{B1}), the general $k$-point correlation function for the eigenvalues
of the random matrix $UA$, $A = {\rm diag} \,(|a|,1,\dots,1)$ with $|a|<1$ was calculated by Fyodorov
\cite{Fy01} (see also \cite{FS03}), up to a proportionality. The starting point was the formula for the
joint eigenvalue distribution
$$
(1 - |a|^2)^{1-N}  \delta \Big ( |a|^2 - \prod_{l=1}^N | z_l|^2 \Big ) \prod_{l=1}^N \chi_{|z_l|<1}
\prod_{1 \le j < k \le N} | z_j - z_k|^2,
$$
where $\chi_J = 1$ for $J$ true and 0 otherwise, valid up to an $N$ dependent normalisation.

To present the result for the corresponding $k$-point correlation, define
\begin{equation}\label{q}
q_j(z_1,\dots,z_k) = [s^j] \det \Big [ \Big (s + x {d \over dx} \Big ) {x^N - 1 \over x - 1} \Big |_{x = z_i \bar{z}_j}
\Big ]_{i,j=1}^k,
\end{equation}
where $[s^j]$ denotes the coefficient of $s^j$ in the expression that follows.
Then, up to a proportionality $c_{N,k}$ say, which may depend on $N,k$,  we have from \cite{Fy01,FS03} that
\begin{multline}\label{12}
\rho_{(k)}(z_1,\dots,z_k) = c_{N,k} (1 - |a|^2)^{1-N}
\chi_{\prod_{l=1}^k | z_l|^2 \ge |a|^2} \prod_{l=1}^k \chi_{|z_l| < 1} \\
\times
\sum_{l=0}^k q_l(z_1,\dots,z_k) \Big ( {d \over dx} x \Big )^l
\Big [ {1 \over x} \Big ( 1 - {|a|^2 \over x} \Big )^{N-1} \Big ] \Big |_{x = \prod_{l=1}^k |z_l|^2}.
\end{multline}
To determine $c_{N,k}$ we consider the case $a=0$. The only nonzero term in 
(\ref{12}) is then $l=0$, telling us that
\begin{equation}\label{qA}
\rho_{(k)}(z_1,\dots,z_k) \Big |_{a=0} = c_{N,k}
\prod_{l=1}^k \chi_{|z_l| < 1} 
\det \Big [  {d \over dx} {x^N - 1 \over x - 1} \Big |_{x = z_i \bar{z}_j} \Big ].
\end{equation}
This is the result for the eigenvalue $k$-point correlation of the ensemble formed
by deleting one row and one column from $U \in U(N)$ chosen with Haar measure,
first derived in \cite{ZS99}. Moreover, the latter contains the explicit form of the proportionality,
which we read off to be
\begin{equation}\label{q1}
c_{N,k} = {1 \over \pi^k},
\end{equation}
and in particular is independent of $N$.
In the case $k=1$ we can check that with (\ref{q1}) substituted in (\ref{12}), (\ref{B1})
is reclaimed.

Our specific interest is in the limiting form of (\ref{12}) with $a = 1/(\mu \sqrt{N})$.

\begin{prop}\label{P2}
Consider (\ref{12}) with the substitution (\ref{q}), and scale $a$ to depend on $N$ according to
$a = 1/(\mu \sqrt{N})$. Let
\begin{equation}\label{Q}
Q_j(z_1,\dots,z_k) = [s^j] \det \Big [ \Big (s + x {d \over dx} \Big ) {1 \over 1 - x} \Big |_{x = z_i \bar{z}_j}
\Big ]_{i,j=1}^k
\end{equation}
(cf.~(\ref{q})).
In the limit $N \to \infty$ we have
\begin{equation}\label{R2a}
\rho_{(k)}(z_1,\dots,z_k) = { e^{1 / |\mu|^2}  \over \pi^k} 
\sum_{l=0}^k Q_l(z_1,\dots,z_k) \Big ( {d \over dx} x \Big )^l
\Big ( {1 \over x} e^{- 1/ ( |\mu|^2 x )} \Big ) \Big |_{x = \prod_{l=1}^k |z_l|^2}.
\end{equation}
restricted to $|z_j|<1$ ($j=1,\dots,k$).
\end{prop}

\begin{proof}
This is immediate from (\ref{12}), and the elementary limit $(1-u/N)^N \to e^{-u}$ as
$N \to \infty$

\end{proof}

Note that (\ref{R2a}) does not correspond to a determinantal point process, as it is not of
the form $\det [K(x_i,x_j)]_{i,j=1}^k$ for some function $K(x,y)$ independent of $k$.
However it does approach a determinantal point process in the limit $|\mu| \to \infty$.
Then, only the term $l=0$ contributes to the sum in (\ref{R2a}), showing that
\begin{equation}\label{R3}
\lim_{|\mu| \to \infty} \rho_{(k)}(z_1,\dots,z_k) = {1 \over \pi^k} \det
\Big [ {d \over dx}  {1 \over 1 - x} \Big |_{x = z_i \bar{z}_j}
\Big ]_{i,j=1}^k
\end{equation}
This is just the result (\ref{1}) for the correlations of the zeros of the limiting random complex
Kac polynomial, as is consistent with $\{z_j\}$, before the limit $\mu \to \infty$,
 corresponding to the solutions of the
equation $1/\mu = \sum_{j=1}^\infty c_j z^j$ with $c_j$ standard complex Gaussians.

In the case $k=1$, (\ref{R2a}) reduces to (\ref{R1}). For future reference, we note too that
for $k=2$, (\ref{R2a}) gives
\begin{equation}\label{L1}
\rho_{(2)}(z,w) = {e^{(1/|\mu|^2)(1 - 1/ \prod_{l=1}^2 | zw|^2)}  \over \pi^2} \bigg ( {1 \over  | zw|^2} Y_0(z,w) 
+ {Y_1(z,w) - Y_2(z,w) \over |\mu|^2 |zw|^4} + {Y_2(z,w) \over |\mu|^4 |zw|^6} \bigg ),
\end{equation}
with
\begin{align}\label{L2}
Y_0(z,w) & =  |zw|^2 \bigg ( {1 \over (1 - |z|^2)^2}{1 \over (1 - |w|^2)^2} -  {1 \over (1 - z \bar{w})^2} {1 \over (1 - \bar{z} w)^2}  \bigg )
\nonumber  \\
Y_1(z,w) & = {|w|^2 \over (1 - |z|^2)(1 - |w|^2)^2} +{ |z|^2 \over (1 - |z|^2)^2(1 - |w|^2)} -
{1 \over |1 - z \bar{w}|^2} \Big ( {w \bar{z} \over 1 - w \bar{z}} + {z \bar{w} \over 1 - z \bar{w}} \Big ) \nonumber \\
Y_2(z,w) & ={|z-w|^2 \over (1 - |z|^2) (1 - |w|^2) |1 - z \bar{w} |^2}
\end{align}
(in obtaining the formula for $Q_2$, use has been made of the Cauchy double alternant determinant formula;
see e.g.~\cite[eq,~(4.34)]{Fo10}).

\begin{remark}
Generally the truncated (or connected) $k$-point correlation is defined by the formula
\begin{equation}\label{T1}
\rho_{(k)}^T(x_1,\dots,x_k) = \sum_{m=1}^k \sum_G (-1)^{m-1} (m-1)! \prod_{j=1}^m \rho_{(|G_j|)}(x_{g_j(1)},
\dots, x_{g_j(|G_j|)}),
\end{equation}
where the sum over $G$ is over all subdivisions of $\{1,2,\dots,k\}$ into $m$ subset $G_1,\dots,G_m$,
with $G_j = \{g_j(1),\dots, g_j(|G_j|) \}$. For example
\begin{align*}
\rho_{(2)}^T(x,y) & = \rho_{(2)}(x,y) - \rho_{(1)}(x) \rho_{(1)}(y) \\
\rho_{(3)}^T(x,y,z) & = \rho_{(3)}(x,y,z) - \rho_{(1)}(x) \rho_{(2)}(y,z)  - \rho_{(1)}(y) \rho_{(2)}(x,z) \\
& \quad -  \rho_{(1)}(z) \rho_{(2)}(x,y) + 2  \rho_{(1)}(x) \rho_{(1)}(y) \rho_{(1)}(z) .
\end{align*}
In the case of a determinantal point process with correlation kernel $K(x,y)$, 
\begin{equation}\label{T2}
\rho_{(k)}^T(x_1,\dots,x_k) = (-1)^{k-1} 
\sum_{{\rm cycles} \atop {\rm length} \,k}
K_N(x_{i_1},x_{i_2}) K_N(x_{i_2},x_{i_3}) \cdots
K_N(x_{i_k},x_{i_1}),
\end{equation}
where the sum is over all distinct cycles $i_1 \to i_2 \to \cdots \to
i_k \to i_1$ of $\{1,\dots,k\}$ which are of length $k$ (see e.g.~\cite[Prop.~5.1.2]{Fo10}). Moreover, if the correlation kernel
has the reproducing property
$$
\int_{\Omega} K(w_1,z) K(z,w_2) \, dz^{\rm r} dz^{\rm i} = K(w_1,w_2),
$$
(here $z^{\rm r}, z^{\rm i}$ denotes the real and imaginary parts of $z$)
as is true in the case of (\ref{1}) for example with $\Omega = D_1$ the domain $|z| < 1$,  then
the structure (\ref{T2}) implies the fundamental sum rule
\begin{equation}\label{T3}
\int_{D_1} \rho_{(k+1)}^T(z_1,\dots,z_k,z) \, dz^{\rm r} dz^{\rm i} = - k  \rho_{(k)}^T(z_1,\dots,z_k)
\end{equation}
(see e.g.~\cite[eq.~(14.12)]{Fo10}).

One can ask if (\ref{T3}) holds in the case of (\ref{R2a}). Not being
determinantal, this appears difficult to answer, even in the case of $k=1$ when we have the explicit functional forms
(\ref{L1}) and (\ref{R1}).
\end{remark}

\begin{remark}
The zero with the smallest modulus is a natural observable quantity. Let $E(0;r)$ denote the probability that the
smallest modulus is greater than or equal to $r$, and denote by $D_r$ the disk of radius $r$.
We know that $E(0;r)$ is given in terms of the correlation functions by (see e.g.~\cite[eq.~(9.4)]{Fo10})
\begin{equation}\label{E}
E(0;r) = 1 + \sum_{k=1}^\infty {(-1)^k \over k!} \int_{(D_r)^k} dz_1^{\rm r} dz_1^{\rm i} \cdots
dz_k^{\rm r} dz_k^{\rm i} \, \rho_{(k)}(z_1,\dots,z_k).
\end{equation}
In the case $|\mu| \to \infty$, when $\rho_{(k)}$ is a $k \times k$ determinant with correlation kernel (\ref{1}),
$E(0;r)$ is given by the Fredholm determinant (see e.g.~\cite[eq.~(9.15)]{Fo10})
\begin{equation}\label{EK}
E(0;r) = \det (  I - \mathbb K_r),
\end{equation}
where $\mathbb K_r$ is the integral operator on $D_r$ with kernel (\ref{1}). The nonzero eigenvalues of the latter are
$\{(2p+2)\}_{p=0}^\infty$, corresponding to the eigenfunctions $\bar{z}^p$ and so \cite{PV03} (see also \cite{Bu18})
\begin{equation}\label{E2}
E(0;r) \Big |_{| \mu | \to \infty} = \prod_{p=0}^\infty (1 - r^{2p+2}).
\end{equation}

For general $| \mu |$ there is no analogue of (\ref{E2}), but some asymptotic results are possible. First, in the limit
$r \to 0^+$, it follows from (\ref{E}) and (\ref{R1}) that
\begin{equation}\label{EK1}
E(0;r) \mathop{\sim}\limits_{r \to 0^+} 1 - {2 \over |\mu|^2} \int_0^r s^{-3} e^{-1/s^2|\mu|^2} \, ds = 1 - e^{-1/(r | \mu|)^2}.
\end{equation}
Second, for $r \to 1^-$, in keeping with the final sentence of Remark \ref{Re1}, we might expect that the leading form of
$E(0;r)$ is independent of $|\mu|$, and is thus that implied by (\ref{E2}) \cite{PV03},
\begin{equation}\label{EK1}
E(0;r)  \mathop{\sim}\limits_{r \to 1^-}  e^{-\pi^2/(12(1-r))}.
\end{equation}
\end{remark}

\section{One and two-point correlations for the zeros of the limiting Kac polynomial}\label{S3}
With $z \mapsto 1/z$, the Laurent series (\ref{L1}) becomes a Maclaurin series.
Through the relation of Proposition \ref{P1}, the results of the previous section give
the correlation functions for the equation
\begin{equation}\label{L2a}
0 =  {1\over \mu} - \sum_{j=1}^\infty {c_j z^j},
 \end{equation} 
 with each $c_j$ is a standard complex Gaussian.
Truncating the latter at the $N$-th term gives the random polynomial
\begin{equation}\label{L2}
 {1\over \mu} - \sum_{j=1}^N {c_j z^j}.
 \end{equation} 
Defining now $p_N(z) = \sum_{j=1}^N {c_j z^j}$,
 our interest is this section is in the statistical properties of the
 solutions of the polynomial equation 
 \begin{equation}\label{L2b}
 p_N(z) = 1/\mu.  
  \end{equation} 
 Taking the limit
 $N \to \infty$ must then reclaim the results of the previous section. We
 are able to carry out this program for the one and two-point correlations.
 
 With $c_j^{\rm r}$ and $c_j^{\rm i}$ denoting the real and imaginary parts of $c_j$ respectively. The probability
 measure of the coefficients  of $p_N(z)$ is then
 \begin{equation}\label{11.92a}
 \Big ( {1 \over \pi}  \Big )^N \exp \Big ( - \sum_{j=1}^N  |c_j|^2 \Big )
 \prod_{j=1}^N  d c_j^{\rm r} d c_j^{\rm i}.
 \end{equation}
 We take up the task of specifying the $k$-point correlation function for solutions of (\ref{L2b}).
 
 Following \cite{Ha96},
 the first step is to introduce complex numbers $z_1^{(0)}, \dots,
z_k^{(0)}$ $(k < (N+1)/2)$ and to define $2k$ linear combinations
of the coefficients $\{ c_j \}_{j=1}^N$ by
$p(z_l^{(0)}) =: p_l$ and $p'(z_l^{(0)}) =: p_l'$ 
$(l=1,\dots,k)$, where the prime denotes differentiation. General
properties of the Gaussian distribution (see e.g.~\cite[Exercises 15.3 q.3]{Fo10}
give that in terms of these $2k$ complex variables (with the
other $N+1-2k$ complex variables integrated out), the probability
measure (\ref{11.92a}) reduces to
\begin{equation}\label{A}
\Big ( {1 \over  \pi} \Big )^{2k} {1 \over \det M}
\exp \Big ( -  (\mathbf{p}, \mathbf{p}{\, '})^\dagger  M^{-1}
(\mathbf{p}, \mathbf{p}\,') \Big ) \prod_{l=1}^k dp_l^{\rm r} dp_l^{\rm i}
dp_l'^{\rm r} dp_l'^{\rm i},
\end{equation}
where $M$ is the covariance matrix
\begin{equation}\label{ABC}
 M = \left [ \begin{array}{cc}
A &  B \\  B^\dagger &  C \end{array} \right ],
\: \:  A := [\langle p_j \bar{p}_l \rangle]_{j,l=1,\dots,k}, \: \:
 B := [\langle p_j \bar{p}_l' \rangle]_{j,l=1,\dots,k}, \: \:
C := [\langle p_j' p_l' \rangle]_{j,l=1,\dots,k}.
\end{equation}
The averages specifying the matrix elements of $A,B,C$ are all with
respect to the coefficients $\{ c_j \}_{j=1}^N$.

For given points $\{z_j^{(0)}\}_{j=1}^k$ (\ref{A})
gives the p.d.f.~for the corresponding values of
$\{p_j\}_{j=1}^k$ and 
$\{p'_j \}_{j=1}^k$. The next step is to change variables so that
the points $\{z_j^{(0)}\}_{j=1}^k$ replace the function values
$\{p_j \}_{j=1}^k$ as the variables. Since $p(z) = p_j + (z - z_l^{(0)}) p_j' +
{\rm O}(( (z - z_l^{(0)})^2)$, the
Jacobian for each such
change of variables equals $|p_l'|^2$. Hence (\ref{A})
transforms to
$$
\Big ( {1 \over  \pi} \Big )^{2k}    {1 \over \det  M}
\exp \Big ( -   (\mathbf{p}, \mathbf{p}{\, '})^\dagger { M}^{-1}
(\mathbf{p}, \mathbf{p}\,') \Big ) \prod_{l=1}^k |p_l'|^2
dz_l^{\rm r} dz_l^{\rm i}
{dp_l'}^{\rm r} {dp_l'}^{\rm i}.
$$

A crucial point is that this last 
change of variables is only locally one to one, as there
will in general be $N$ points giving the same function value. This
shows that if we set $\mathbf{p} = {1 \over \mu} \mathbf{1}_k$,  and integrate over
each $p_l'^{\rm r}, p_l'^{\rm i}$ the $k$-point correlation function for the
$N$ complex solutions of $p_N(z) = 1/\mu$
will result, giving
\begin{multline}\label{11.zero}
\rho_{(k)}(z_1^{(0)}, \dots, z_k^{(0)}) = \\
\Big ( {1 \over  \pi } \Big )^{2k} {1 \over \det  M}
 \int_{(-\infty,\infty)^{2k}} 
\prod_{l=1}^k
dp_l'^{\rm r} dp_l'^{\rm i} \,
 |p_l'|^2 \exp \Big ( -   ({1 \over \mu} \mathbf{1}_k, \mathbf{p}{\, '})^\dagger { M}^{-1}
({1 \over \mu} \mathbf{1}_k, \mathbf{p}\, ') \Big ),
\end{multline}
where $({1 \over \mu} \mathbf{1}_k, \mathbf{p}\, ')$ denotes the $2k \times 1$ column vector
obtained by concatinating $ \mathbf{1}_k$ and $ \mathbf{p}'$
(cf.~\cite[eq.~(15.56)]{Fo10}). This can be viewed as a generalisation of the Kac-Rice formula; see e.g.~\cite{AT07}.

In the case $1/\mu = 0$ the multiple integral (\ref{11.zero}) can be written in terms of a permanent \cite{Ha96},
although the underlying matrix depends on $k$, making its explicit evaluation difficult except for small $k$
(or, as it turns out, large $N$). The permanent structure breaks down when $1/\mu \ne 0$, but nonetheless 
exact computation is still possible for small $k$; we restrict attention to $k=1$ and $k=2$.
Analogous to the strategy for the case $1/\mu = 0$,
we begin by eliminating the integrations in (\ref{11.zero}) in favour of differentiations.

\begin{prop}
Write
\begin{equation}\label{NN}
 M^{-1} = \left [ \begin{array}{cc}
N_1 &  N_2 \\  N_2^\dagger &  N_3 \end{array} \right ],
\end{equation}
where each $N_i$ is of size $k \times k$ (cf.~(\ref{ABC})). Introduce the auxiliary vector $\mathbf{\nu} = (\nu_1,\dots, \nu_k)$,
where each $\nu_j$ is a complex number, and let $\bar{\pmb \nu}$ denote its complex conjugate.

The multiple integral formula (\ref{11.zero}) can be rewritten
\begin{multline}\label{SM1}
\rho_{(k)}(z_1^{(0)}, \dots, z_k^{(0)}) = \Big ( {1 \over  \pi} \Big )^k {1 \over \det A} e^{- {1 \over \overline{\mu}} \mathbf{1}_k^T N_1 {1 \over \mu} \mathbf{1}_k}
{\partial^k \over \partial \bar{\nu}_1 \cdots  \partial \bar{\nu}_k} 
\prod_{l=1}^k \Big (  (  \bar{\pmb \nu} + i {1 \over \bar{\mu}} {\mathbf 1}_k^T N_2) \cdot N_3^{-1} \mathbf{e}_l \Big ) \\
\times \exp \Big (  (  \bar{\pmb \nu} + i {1 \over \overline{\mu}} \mathbf{1}_k^T N_2) \cdot N_3^{-1} ({\pmb \nu} - i N_2^\dagger {1 \over \mu}  \mathbf{1}_k \Big ) \Big |_{\mathbf{\nu} = \mathbf{0}}.
\end{multline}
\end{prop}

\begin{proof}
Our initial reason to introduce differentiations is to eliminate
the factor $\prod_{l=1}^k | p_l'|^2$ in the integrand of (\ref{11.zero}).
Thus we observe
\begin{multline}\label{SM}
\rho_{(k)}(z_1^{(0)}, \dots, z_k^{(0)}) = \Big ( {1 \over \pi} \Big )^{2k} {1 \over \det M} 
{\partial^{2k} \over  \partial {\nu}_1 \cdots   \partial {\nu}_k  \partial \bar{\nu}_1 \cdots \partial \bar{\nu}_k} \\
\times
\int_{(-\infty,\infty)^{2k}} 
 \exp \Big ( -   ({1 \over \mu} \mathbf 1_k, \mathbf{p}{\, '})^\dagger { M}^{-1}
({1 \over \mu} \mathbf{1}_k, \mathbf{p}\, ')  +  i \overline{\pmb  \nu} \cdot \mathbf{p}' -  i {\pmb \nu} \cdot \mathbf{p}'{}^* \Big ) \,
\prod_{l=1}^k
dp_l'^{\rm r} dp_l'^{\rm i} 
\Big |_{{\pmb \nu} = \mathbf{0}},
\end{multline}
where here $\mathbf{p}'{}^*$ denotes the complex conjugate of the vector $\mathbf{p}'$.
In terms of the notation (\ref{NN}) we can expand
$$
({1 \over \mu} \mathbf 1_k, \mathbf{p}{\, '})^\dagger { M}^{-1}
({1 \over \mu} \mathbf{1}_k, \mathbf{p}\, ')
  = ({1 \over \bar{\mu}} \mathbf{1}_k)^TN_1 {1 \over \mu} \mathbf{1}_k +
(\mathbf{p}')^\dagger N_3 \mathbf{p}' + i (-i {1 \over \bar{\mu}} \mathbf{1}_k N_2 \cdot \mathbf{p}') +
i (-i N_2^\dagger {1 \over \mu} \mathbf{1}_k) \cdot (\mathbf{p}')^*.
$$
Substituting this in (\ref{SM}) shows that the integrand is the exponential of a quadratic form in
$\{ dp_l'^{\rm r},  dp_l'^{\rm i} \}_{l=1}^k$, allowing the integral to be evaluated (see e.g.~\cite[eq.~(15.65)]{Fo10}).
Once this is done, the differentiations over $\{ \nu_j \}_{j=1}^k$ can be carried out.
Finally, using the fact that $\det M \det N_3 = \det A$, we arrive at (\ref{SM1}).
\end{proof}

Starting from (\ref{SM1}), it is straightforward to obtain explicit formulas for
$\rho_{(1)}$ and $\rho_{(2)}$. There are further simplifications in the limit
$N \to \infty$, which is the case of interest for purposes of the present study.
We begin by considering $\rho_{(1)}$.

\begin{prop}
In the limit $N \to \infty$, and requiring that $|z| < 1$,  $\rho_{(1)}(z)$ is given by (\ref{R1}).
\end{prop}

\begin{proof}
In the case $k=1$ the quantities $\pmb \nu, \pmb \mu, A, N_2, N_3$ are then scalars, and
(\ref{SM1}) simplifies to 
\begin{equation}\label{ABCS}
\rho_{(1)}(z) = {1 \over \pi} {1 \over A N_3} e^{-{N_1 \over |\mu|^2 }+ { |N_2|^2 \over |\mu|^2 N_3} }\Big ( 1 +
 {|N_2|^2 \over |\mu|^2 N_3} \Big ).
\end{equation}
Since $M$ is of size $2 \times 2$, its inverse is simple to compute, and in terms of (the scalars)
$A,B,C$ we read off that
$$
N_1 = {C \over AC - |B|^2}, \quad N_2 = -{B \over AC - |B|^2}, \quad
N_3 = {A \over AC - |B|^2}.
$$
Moreover, it follows from the definition (\ref{ABC}) that in the limit $N \to \infty$
\begin{equation}\label{ABCn}
A = {|z|^2 \over 1 - |z|^2}, \quad B = {z \over (1 - z \bar{z})^2}, \quad
C = {1 + z \bar{z} \over (1 - z \bar{z})^3}.
\end{equation}
Thus
$$
AC - |B|^2 = {|z|^4 \over (1 - |z|^2)^4}, \quad N_1 = {1 \over |z|^4}(1 - |z|^4), \quad
{ |N_2|^2 \over N_3} = {1 \over |z|^4} (1 - |z|^2).
$$
Substituting the above in (\ref{ABCS}) gives (\ref{R1}).
\end{proof}

The difficulty with computing $\rho_{(k)}$ for $k \ge 2$ is that $A,B,C$ are no longer
scalars, and correspondingly it is not immediate that there are structured formulas for
the matrix elements of $N_1,N_2,N_3$. In the case $k=2$, the required computations
can be carried out using computer algebra.

\begin{prop}
In the limit $N \to \infty$, and with $|w|, |z| < 1$,
\begin{multline}\label{Z2}
\rho_{(2)}(w, z) =  \Big ( {1 \over \pi}  \Big )^2 
 \exp \bigg ( {1 \over  |\mu|^2} \Big ( 1- {1 \over |z|^2 |w|^2} \Big ) \bigg )
 \\
\times {| z - w|^2 \over | 1 - w \bar{z} |^2} {1 \over |w|^2 |z|^2 (1 - |z|^2) (1 - |w|^2)} 
\bigg ( {1 \over  |\mu|^4 |w|^4 |z|^4} \\
+ {1 \over |\mu|^2} \Big ( {1 \over (1 - |w|^2)^2 |z|^2} + {1 \over (1 - |z|^2)^2 |w|^2} +
{1 \over w \bar{z} (1 - z \bar{w})^2 } +  {1 \over z \bar{w} (1 - w \bar{z})^2 } \Big )\\
+ |w|^2 |z|^2 \Big (
{1 \over (1 - |w|^2)^2   (1 - |z|^2)^2 } + {1 \over (1 - z \bar{w})^2   (1 - \bar{z} w)^2} \Big ) \bigg ).
\end{multline}
\end{prop}

\begin{proof}
Introduce the notation
$$
 i {1 \over \bar{\mu}} \mathbf 1_2 N_2 \cdot N_3^{-1} \mathbf e_1 =: G_1, \quad
 i {1 \over \bar{\mu}} \mathbf 1_2 N_2 \cdot N_3^{-1} \mathbf e_2 =: G_2, \quad
h_{jk} =  \mathbf e_j \cdot N_3^{-1} \mathbf e_k.
$$
Performing the differentiations as required by (\ref{SM1}) shows, upon minor manipulation,
that 
\begin{multline}\label{ER}
\rho_{(2)}(z, w) = {1 \over \pi^2} {1 \over \det A} 
 e^{-{1 \over |\mu|^2 } {\mathbf 1}_2^T N_1 {\mathbf 1}_2} e^{ {1 \over |\mu|^2 } (  {\mathbf 1}_2^T N_2) \cdot N_3^{-1} (N_2^\dagger  {\mathbf 1}_2)}
\\
\times \Big ( (|G_1|^2 + h_{11}) (|G_2|^2 + h_{22}) +  (G_1 \overline{G}_2  + \overline{h}_{12}) (\overline{G}_1 G_2 + h_{12}) - |G_1 G_2|^2 \Big ).
\end{multline}

The $k=2$ modification of (\ref{ABCn}) is
\begin{equation}\label{ABCs}
A = \Big [ {z_i \overline{z}_j  \over 1 - z_i \overline{z}_j} \Big ]_{i,j=1,2}, \quad B =   \Big [ {z_i \over (1 - z_i \bar{z}_j)^2}  \Big ]_{i,j=1,2} , \quad
C =   \Big [ {1 + z_i \bar{z}_j \over (1 - z_i \bar{z}_j)^3}  \Big ]_{i,j=1,2},
\end{equation}
where $z_1 = z$ and $z_2 = w$. Use of the(see e.g.~\cite[eq,~(4.34)]{Fo10}) known Cauchy double alternant determinant formula  gives
$$
{1 \over \det A} = {(1 - |w|^2)  (1 - |z|^2)  (1 - z \bar{w}) (1 - w \bar{z})  \over |z|^2 |w|^2 |w - z|^2}
$$
For the other quantities in (\ref{ER}), beginning with (\ref{ABCs}), we make use of computer algebra to compute
\begin{align*}
{\mathbf 1}_2^T N_1 {\mathbf 1}_2 & = \Big ( - 1 + {1 \over |w|^4 |z|^4} \Big ) \\
( {\mathbf 1}_2^T N_2) \cdot N_3^{-1} (N_2^\dagger  {\mathbf 1}_2) & = \Big (  {1 \over |w|^4 |z|^4} -  {1 \over |w|^2 |z|^2} \Big ) \\
 G_1 & = {i \over \bar{ \mu}} \Big (  {\bar{z} - \bar{w} \over
 (1 - |z|^2) (1 - w \bar{z}) \bar{w} \bar{z} } \Big ) \\
 h_{11} & = {|z|^2 | z - w|^2 \over (1 - |z|^2)^3  (1 - z \overline{w}) (1 - w \overline{z})} \\
 h_{12} &= - {z \overline{w} | z - w|^2 \over (1 - |w|^2) (1 - |z|^2) (1 - z \overline{w})^3}
 \end{align*}
 as well as $G_2 = G_1 |_{w \leftrightarrow z}$, $h_{22} = h_{11} |_{w \leftrightarrow z}$,
 $h_{21} = \overline{h}_{12} = h_{21}  |_{w \leftrightarrow z}$. Substituting in (\ref{ER}) and simplifying gives (\ref{Z2}).

\end{proof}

It remains to verify that (\ref{Z2}) agrees with (\ref{L1}). In the notation of (\ref{L1}), this requires checking that
\begin{align}
Y_0(z,w) & = {|z-w|^2 \over | 1 - w \bar{z} |^2} { |w|^2 |z|^2 \over (1 - |z|^2) (1 - |w|^2)} \nonumber \\
& \qquad \times \bigg ( {1 \over (|1 - |w|^2)^2 (1 - |z|^2)^2} +
{1 \over (|1 - z \overline{w})^2 (1 - \bar{z} w )^2} \bigg ) \label{Y0} \\
Y_1(z,w) - Y_2(z,w) & =
{ |z - w|^2 \over |1 - w\bar{z}|^2} {|wz|^2 \over (1 - |z|^2)(1 - |w|^2)} \nonumber \\
& \quad \times
\bigg ( {1 \over (1 - |w|^2)^2 |z|^2} + {1 \over (1 - |z|^2)^2 |w|^2} +
{1 \over w \bar{z} (1 - z \bar{w})^2} +
{1 \over z \bar{w} (1 - w \bar{z})^2}  \Big ) \label{Y1} \\
Y_2(z,w) & = {|z-w|^2 \over |1 - w \bar{z} |^2 (1 - |z|^2) (1 - |w|^2)} \label{Y2}
\end{align}
Comparing with (\ref{L2}), we see that (\ref{Y2}) is satisfied by definition, while (\ref{Y0}) is equivalent to the
Borchardt identity (see e.g.~\cite{Si04})
$$
\det \Big [ {1 \over (1 - x_j y_k)^2} \Big ]_{j,k=1}^N =
{\prod_{1 \le j < k \le N} (x_k - x_j) (y_k - y_j) \over \prod_{j,k=1}^N (1 - x_j y_k)}
{\rm perm} \, \Big [ {1 \over 1 - x_j y_k} \Big ]_{j,k=1}^N
$$
in the case $N=2$ with $x_1=z, y_1=\bar{z}, x_2 = w, y_2 = \bar{w}$. While we don't know of any structured explanation
of (\ref{Y1}), it can be recast as a polynomial identity in $z,\bar{z},w,\bar{w}$ upon clearing denominators, and verified by
computer algebra.

\begin{remark} The form (\ref{Z2}) contains the factor $|z-w|^2$, showing a repulsion between near degenerate
eigenvalues of the same form as that for complex Gaussian matrices (see e.g.~\cite[\S 15.1]{Fo10}).
\end{remark}

\section*{Acknowledgements}
The work is part of a research program supported by the Australian Research Council 
Centre of Excellence for Mathematical and Statistical Frontiers. PJF also acknowledges
partial support from the Australian Research Council 
Grant DP170102028.

%\bibliographystyle{amsplain}
%\bibliography{book1}

%input{KacPoly.bbl}
\providecommand{\bysame}{\leavevmode\hbox to3em{\hrulefill}\thinspace}
\providecommand{\MR}{\relax\ifhmode\unskip\space\fi MR }
% \MRhref is called by the amsart/book/proc definition of \MR.
\providecommand{\MRhref}[2]{%
  \href{http://www.ams.org/mathscinet-getitem?mr=#1}{#2}
}
\providecommand{\href}[2]{#2}

  \end{document}